\documentclass[12pt]{article}

\usepackage{epsfig}

\usepackage{amssymb}
\usepackage{amsfonts}

\usepackage{color}
\def\be{\begin{equation}}
\def\ee{\end{equation}}
\def\ba{\begin{array}{c}}
\def\ea{\end{array}}

\def\ben{$$}
\def\een{$$}

\newcommand{\bea}{\begin{eqnarray}}
\newcommand{\eea}{\end{eqnarray}}

\newcommand{\kt}{\rangle}

\newtheorem{thm}{Theorem}

\newtheorem{lemma}[thm]{Lemma}

\newtheorem{defn}[thm]{Definition}

\newenvironment{proof}{\noindent
 {\bf Proof.}}{\hfill$\square$\vspace{3mm}\endtrivlist}

\begin{document}

\begin{center}

{\Large \bf

Non-Hermitian
$N-$state
degeneracies: unitary
realizations via
antisymmetric anharmonicities

}

\vspace{0.8cm}

  {\bf Miloslav Znojil}

\vspace{0.2cm}

The Czech Academy of Sciences, Nuclear Physics Institute,

 Hlavn\'{\i} 130,
250 68 \v{R}e\v{z}, Czech Republic

\vspace{0.2cm}

 and

\vspace{0.2cm}

Department of Physics, Faculty of Science, University of Hradec
Kr\'{a}lov\'{e},

Rokitansk\'{e}ho 62, 50003 Hradec Kr\'{a}lov\'{e},
 Czech Republic

\vspace{0.2cm}

{e-mail: znojil@ujf.cas.cz}

\end{center}

%\end{document}

%\newpage

\section*{Abstract}

In a way inspired by the phenomenon of
quantum phase transitions
two mutually interrelated aspects of their theory
are addressed. The first one is pragmatic:
a non-numerical picture
of the unitary
processes of the loss of the observability
is offered. A specific anharmonic-oscillator (AHO) class of the
real $N$ by $N$ matrix toy-model Hamiltonians
is recalled and adapted for the purpose.
The second aspect is theoretical:
the non-Hermitian degeneracy of
an $N-$plet of the stable bound states
(a.k.a.
the Kato's exceptional-point - EPN - degeneracy)
is realized in the language of unitary
(called, sometimes, ${\cal PT}-$symmetric)
quantum physics of closed systems.
What is obtained
is a classification
pattern involving
an auxiliary integer $K$
(prescribing a ``clusterization index''
{\it alias\,} geometric multiplicity of
the EPN limit) and the choice of one of the partitionings
of the equidistant unperturbed spectrum
into equidistant and centered unperturbed
subspectra.

\subsection*{Acknowledgments}

The author acknowledges the financial support from the
Excellence project P\v{r}F UHK 2020.

\newpage

\newpage

\section{Introduction}

The practical, experiment-oriented
descriptions of processes of
quantum phase transitions
mediated by non-Hermitian degeneracies \cite{Berry}
{\it alias\,} exceptional points (EP, \cite{Kato})
are usually based,
in the Feshbach's
open-system spirit \cite{Feshbach},
on an {\it ad hoc\,}
effective Hamiltonian $H_{\rm eff}$.
In these applications
of the idea the systems in question are usually
next-to-prohibitively complicated
so that their
analysis usually requires various forms of
a simplification
of mathematics.
For this reason the criteria of the
trial-and-error choice of $H_{\rm eff}$
are mostly purely pragmatic and not too deeply
motivated.
A more ambitious
theoretical background, if any,
is
provided {\it a posteriori}.
Typically,
the necessary omission of
all of the irrelevant degrees of freedom
is made via perturbation-theory-based tests of their
negligibility \cite{Nimrod}.

The dominance of the underlying open-system
philosophy has been shattered by Bender and Boettcher \cite{BB}.
They proposed that
at least some of the processes of the
quantum phase transitions
(and, in particular, the
loss-of-the-observability processes)
might find a more natural
description and explanation in the alternative,
closed-system theoretical
framework (see some of its details in reviews
\cite{Carl,ali,book}).
The Bender- and
Boettcher-inspired change of the paradigm has two roots.
One is that even in many closed
quantum systems the unitary evolution
can be controlled by the Hamiltonian
(with real spectrum)
which is non-Hermitian
(i.e., admitting the EPs).
Although such a conjecture might sound like a paradox,
the Bender's and Boettcher's
secret trick is that the latter Hamiltonian
can be,
via an
appropriate amendment of the Hilbert space
of states,  hermitized,
fitting suddenly all of the standard postulates
of textbooks
(see, e.g., the older review paper \cite{Geyer}
for explanation).

The second root of the change of the paradigm
immediately concerns the problem of the
generic
quantum phase transitions,
opening one of its possible innovative treatments.
The point is that in \cite{BB}
the  phase transitions were
sampled
by the spontaneous breakdown of
parity-time (i.e., ${\cal PT}$) symmetry
of elementary models.
In such a specific class of exemplifications of the
phenomenon
it was not too difficult to relate
the
collapse
of the system
to the coincidence of parameter $\lambda$ in Hamiltonian
$H(\lambda)$ with its
EP
value
$\lambda^{(EP)}$.
The
difference in behavior
between the
open and closed systems proves inessential.
The closed, unitary systems do
encounter, in the
phase-transition limit
$\lambda \to \lambda^{(EP)}$, a
``catastrophe'' realized as a
merger followed by an abrupt
complexification of at least some (or, in general,
of an $N-$plet \cite{arbiorder})
of the
bound-state energies, i.e., by
an abrupt loss of the
observability of the system.

Within the framework of the new
hermitizable-Hamiltonian
paradigm,
one of the most
important ingredients
in the upgraded theory
would be an exhaustive
classification of the
possible new evolution scenarios.
In particular,
any type of such a classification seems needed
for our understanding of
the onset-of-the-phase-transition
dynamical regime.
Unfortunately,
the absence of a similar
classification in the literature reflects the
existence of the numerous
technical obstacles. {\it Pars pro toto\,} let us mention
the scepticism (formulated, e.g., in \cite{bhsmall})
concerning the very feasibility of
analysis of alternative
loss-of-the-observability mechanisms,
or the fairly discouraging results of
several numerical experiments
simulating, in \cite{anomalous},
the non-Hermitian $N-$level degeneracies
at the values of $N$ as small as $N= 6$.

In our present paper we are going to oppose the
scepticism. Our return to optimism has
several independent roots. The most important one is that
we imagined that
there exist enough toy-model matrices which
could help forming certain building blocks of
a sufficiently
rich descriptive, qualitative theory.
In this sense we will construct
a family of
non-numerical benchmark models exhibiting
the explicit
loss-of-the-observability behavior.

The presentation of our results will start, in section
\ref{specia}, by an identification of
an arbitrary non-Hermitian degeneracy of $N$ levels with
a sudden loss of the diagonalizability of the Hamiltonian.
The main goal of our paper may be formulated as a classification
of these EPs of the $N-$th order (or, more explicitly, EPNs)
which would take into account the so called geometric multiplicity $K$
of the $N-$plet of levels in the EPN phase-transition dynamical extreme.

In section \ref{specibe} a key concept of our considerations,
viz., a
certain specific class of anharmonic oscillators is
introduced. In the model
the conventional
matrix form of the (initially dominant) harmonic oscillator
Hamiltonian
is diagonal (with equidistant elements) while the
anharmonicity (i.e., initially, a small perturbation)
is a real matrix which is
antisymmetric. The reasons for such a choice are given.
Subsequently, in section
\ref{specice}, our present systematic search for
the models with anomalous non-Hermitian degeneracies
(i.e., EPNs with $K \neq 1$) is initiated.
We start there from a
five-diagonal $N$ by $N$
matrix ansatz
for $H(\lambda)$,
and we illustrate the emerging picture
of the EPN-related physics
via a specific $N=7$ model.
On a methodical level we finally add and explain the last
technical ingredient of our present analysis, viz.,
the idea of a systematic reduction of the
EPN-supporting Schr\"{o}dinger equations
based on suitable multidiagonal versions of the
anharmonic oscillator Hamiltonians
into a composition of decoupled
(or weakly coupled) sub-equations:
this is done in sections
\ref{specide} and
\ref{specife}
followed by the last section \ref{speciae}
containing a concise discussion
and a brief summary.

\section{Mechanisms of the loss of diagonalizability \label{specia}}

The above-outlined danger of
a non-Hermitian degeneracy and collapse
concerns, in particular,
quantum systems living
in a small vicinity of
their EP
boundary where
$\lambda = \lambda^{(EP)}$.
The current state of our understanding
of these systems
has recently been reviewed
in the immediate predecessor \cite{anomalous}
of our present paper.
First of all
it has been emphasized there that near the
EP-mediated
loss-of-observability
limit
the quantum bound-state problem
becomes,
in a way well known to mathematicians
\cite{Wilkinson,Trefethen},
numerically ill-conditioned.
As a consequence, whenever
$\lambda \approx \lambda^{(EP)}$,
the spectra of the bound-state energies
become extraordinarily
sensitive even to the smallest
random numerical errors
\cite{admissible}.

Fortunately, the
apparent paradox and conceptual problem immediately disappears
when we
recall any formulation of quantum mechanics which is
applicable to the situation (let us recommend, e.g.,
the extensive but well-written review paper \cite{ali}).
For a clarification of the situation
it is necessary to distinguish between the
computer-generated random numerical errors
(which are defined as small in a conventional ``working''
Hilbert space ${\cal K}$)
and
the real-world perturbations which
must be realized and specified as
small  in the physics-representing
Hilbert space ${\cal H}$ \cite{Ruzicka}.

In
paper~\cite{anomalous}
the necessity of a reliable control of numerical
errors
led to the restriction of the scope of the paper to
a few toy-model
matrices with $N=6$.
Even within such a reduced project
the practical control of the numerical precision
remained nontrivial.
In similar situations
one must mostly rely upon
the results obtained by
non-numerical means.

\subsection{Elementary two by two matrix example}

One of the most efficient simplifications
of the implementation of quantum theory
is provided by the
finite-matrix
$N$ by $N$ truncation of the Hamiltonians.
Their finite-dimensional
{\it alias\,} separable-interaction representations prove
methodically useful even at $N=2$.
One of the most persuasive demonstrations
of such a remark may be found
in the Kato's
book on perturbation theory \cite{Kato}.
In it,
the mathematically rigorous definition of the EPs (see
p. 64 in {\it loc. cit.})
is immediately followed by an elementary illustrative example
 \be
 T^{(2)}(\kappa)
 = \left [\begin {array}{cc} 1&\kappa\\{}\kappa&-1
 \end {array}\right ]\,
 \label{kateq}
 \ee
This is a traceless two-by-two matrix
with eigenvalues
$E_\pm(\kappa) = \pm \sqrt{1+\kappa^2}$
[cf. Example 1.1.(a) in {\it loc. cit.}] which
could be perceived as
the simplest possible
truncated version of our present anharemonic-oscillator
models (cf. Eq.~(\ref{pen}) below).
By Kato \cite{Kato}, in contrast, the model
was recalled as yielding the complex conjugate
pair of the EPs
characterized by the
degeneracy of
the levels,
 \be
 \lim_{\kappa \to \kappa^{(EP)}}\,E_+(\kappa)=
 \lim_{\kappa \to \kappa^{(EP)}}\,E_-(\kappa)=0\,,
 \ \ \ \
  \kappa^{(EP)} =  {\rm i} \ \ {\rm and/or}\ \
  \kappa^{(EP)} =  -{\rm i}
  \,.
 \label{sies}
 \ee
In the context of physics of closed, unitary quantum
systems (which will be of primary interest in the present paper)
one has to add that
the benchmark matrix (\ref{kateq}) is Hermitian
(i.e.,
from the point of view of
conventional textbooks \cite{Messiah}, acceptable as a
toy-model Hamiltonian)
if and only if $\kappa$ is real.
From such a perspective,
the unitarity
of the system
becomes lost near all of its EPs.
In model (\ref{kateq}) this makes the EPs much less
interesting for physicists of course.

\subsection{Physics behind EPs}

The inaccessibility of the EPs
is not a model-dependent anomaly. Its validity
extends to
any generic Hermitian $N$ by $N$ matrix
$T^{(N)}(\kappa)$ without {\it ad hoc\,} symmetries.
In the past this made the concept of EPs, from the
phenomenological point of view, useless.
From the present point of view, the scepticism
was undeserved,
contradicting
multiple subsequent
discoveries of
relevance
of the non-Hermitian dynamics and
of the related EPs. Multiple emerging applications
range from the
condensed-matter theory \cite{Dyson} and from the
diverse branches of
experimental physics \cite{Geyer,Heiss}
up to
the
quantum
analogues of the classical theory of
catastrophes \cite{catast}.

Any introductory list of the concrete samples of the existence
and experimental or theoretical relevance of
quantum phenomena connected with
the presence and proximity of the Kato's EPs
in Hamiltonians $H(\lambda)$
at
$\lambda = \lambda^{(EP)}$
should probably start from the recollection of the
EP-related spontaneous breakdown
of the parity-time symmetry
as mentioned by
Bender and Milton \cite{BM}
(in the context of quantum field theory), and by
Bender with Boettcher \cite{BB}
(in a mathematically
much better understood context of quantum
mechanics using non-selfadjoint
operators with the real and discrete
bound-state spectra \cite{ali,book,Geyer}).
At present an updated list of the samples
of the use of EPs in quantum physics would certainly
involve various
quantum phase transitions \cite{denis} and
catastrophes \cite{catast}. What they would all
share is a connection,
direct or indirect, with the formal
limit
 \be
 H(\lambda^{(EP)})=
 \lim_{\lambda \to
 \lambda^{(EP)}} H(\lambda)\,
 \label{esist}
 \ee
of the
Hamiltonian.

In the subset of samples
of the applicability of EPs of our present interest
(dealing, exclusively, with the unitary, {\em closed\,} quantum
systems) the most important mathematical
feature of operator $H(\lambda)$ should be seen in
the reality of its spectrum
(in an ``interval of unitarity''
with $\lambda < \lambda^{(EP)}$),
guaranteed
up to the very EP limit (\ref{esist}).
At the same time,
the limit $H(\lambda^{(EP)})$ itself
cannot be perceived as a valid quantum Hamiltonian anymore.
The main reason is that this operator
is, by definition, non-diagonalizable.
This means that  in
an immediate vicinity of the EP singularity
one should expect the emergence of multiple
interesting
phenomena.

A firm theoretical ground for their description
should be sought in the above-mentioned
amendment ${\cal H}$ of
the conventional Hilbert space ${\cal K}$.
Whenever
the spectrum of any given
$H(\lambda)$ is kept real, such a
Hamiltonian can be non-Hermitian in  ${\cal K}$ but still
Hermitian in the other,
{\it ad hoc\,} Hilbert space ${\cal H}$.
As a consequence,
the evolution generated by the Hamiltonian
may always be interpreted as unitary
(cf., e.g., the thorough review \cite{ali} for details).

The amended Hilbert space will be
$\lambda-$dependent, ${\cal H}={\cal H}(\lambda)$.
Thus, Hamiltonian $H(\lambda)$ itself may be interpreted as
Hermitian in a corridor of $\lambda$s
having a non-empty overlap with an arbitrarily small vicinity
of $\lambda^{(EP)}$ \cite{corridors}. The system is able to reach
the instant of phase transition via unitary evolution.

\subsection{Exceptional point degeneracies and their classification}

Most of the
quantitative analyses
of the EP-influenced dynamical
scenarios have been found technically difficult, especially
when one tries to move beyond the most elementary
models with the smallest matrix dimensions $N$.
There exist several sources of difficulties.
First, for a given non-Hermitian Hamiltonian
$H^{(N)}(\lambda)$,
even the proof of the reality of the spectrum ceases to
be an elementary task.
Even at $N=3$,
the availability of the energies in the exact and
closed form of Cardano formulae is often marred by
a typical occurrence of the mutually canceled complex components
in the formula causing the emergence of
numerical errors \cite{anomalous}.
Second, once we manage to
keep the numerical uncertainties of the energies
under control,
this control becomes more and more difficult
when the parameter $\lambda$ moves closer to its EP
value. Third,
even when we succeed in the precise
localization of the position of the EP
parameter,
we have to keep in mind that
the size of
difficulties will increase with the
growth of dimension $N$,
especially when we study the EPs
of maximal order $N$ (abbreviated as EPNs).

Along all of these stages of development
an efficient help can be provided by
a variable-length computer
arithmetics (see, e.g., \cite{bhgt6}).
Even then,
one still has to complement the necessary condition
of the EPN confluence of {\em all\,} of the $N$ energy levels,
 \be
 \lim_{\kappa \to \kappa^{(EPN)}}\,E_n(\kappa)=\eta\,,\ \ \ \
 n=1,2,\ldots,N
  \,
 \label{siesta}
 \ee
by a more detailed characteristics of the
structure of the parallel confluence of the wave functions.
The general form of the latter confluence will depend on
an integer (say, $K$) called the
geometric multiplicity of the EPN degeneracy~\cite{Kato}.
It will characterize
the following
$K-$centered
clusterization of the eigenstates near the EPN singularity,
 \be
 \lim_{\lambda \to \lambda^{(EPN)}}\,|\psi_{n_k}^{(N)}(\lambda)\kt
 =|\chi_k^{(N)}(\lambda)\kt\,,\ \ \  \ \
 n_k \in S_k\,,
 \ \ \ \
 k=1,2,\ldots,K\,.
 \label{kwinde}
 \ee
In this formula the
$N-$plet of integer subscripts $\{1,2,\ldots,N\}$
counting the states
gets partitioned
into a $K-$plet of its disjoint subsets $S_k$
formed by the $N_k-$plets of the separate indices.
Nontriviality of the situation requires that $N_k\geq 2$
(indeed, the singlets can be ignored as
belonging to an irrelevant, decoupled
part of the Hilbert space).
Finally, once the overall dimension $N$ is fixed,
the separate non-equivalent EPN scenarios
may be numbered
by the set of all possible
partitions
$P[N]=P_j[N]$ of
     $$N=N_1+N_2+\ldots + N_K\,.$$
A short table of all realizations of these partitions
may be found in \cite{without,Acc}.

\section{Antisymmetrically anharmonic oscillators \label{specibe}}

In the area of non-Hermitian quantum mechanics, a
number of methodical challenges
occurred, recently, in connection with the
mathematically rigorous
studies of the
onset of instabilities due to perturbations \cite{Viola}.
In parallel, physicists newly reopened
the questions of
a systematic qualitative understanding
of the EP-related
quantum
phase transitions \cite{BB,denis,denisbe}.
Both of these tendencies in research
are mutually interconnected of course.

\subsection{Conventional anharmonic oscillators}

In the most convenient harmonic-oscillator (HO) basis,
the general anharmonic-oscillator (AHO) Hamiltonian
 \be
 H(\lambda)=H^{(HO)}
 + \lambda\,V_{\rm }\,,
 \ \ \ \ \ \lambda \geq  0\,
 \label{peng}
 \ee
acquires, in the weak-coupling regime, the
diagonally dominated matrix form
 \be
 H(\lambda)
 =\left[ \begin {array}{cccc}
  1+\lambda\,V_{0,0}&\lambda\,V_{0,1}
  &\lambda\,V_{0,2}&\ldots
  \\{}\lambda\,V_{1,0}&3+\lambda\,V_{1,1}
  &\lambda\,V_{1,2}
 &\ddots
 \\{}\lambda\,V_{2,0}&\lambda\,V_{2,1}
 &5+\lambda\,V_{2,2}&\ddots
 \\{}\vdots&\ddots&\ddots&\ddots
 \end {array} \right]\,,\ \ \ \ \
 \lambda =
 {\rm small}
 \,.
 \label{pen}
 \ee
The study of the AHO models was initially motivated
by their methodical perturbation-analysis implications.
The unperturbed
Hamiltonian itself is represented,
in the well known harmonic-oscillator basis,
by the diagonal matrix
with equidistant elements $H^{(HO)}_{0,0}=1$,
$H^{(HO)}_{1,1}=3$, etc.
The couplings $\lambda$
remain, in this setting, small and positive,
$\lambda \in (0,\lambda_{\max})$.

The most popular choice of
the coordinate-dependent
quartic anharmonicity  $V \sim x^4$
converts the Hamiltonian into a
particularly computation-friendly real pentadiagonal matrix.
In \cite{BW}, this encouraged
Bender and Wu to give a detailed account
of the positions of EPs .
The motivation of this early study was still formal,
based on the fact that
one of the most important mathematical prerequisities of
perturbation expansions, viz., a guarantee of convergence
finds a clarification after
an extension
of the range of the admissible
couplings
to
complex plane \cite{Kato}.
The
value of the radius of convergence
of the most common version of the
Rayleigh-Schr\"{o}dinger perturbation series
can be then identified with the distance of the
reference coupling $\lambda_0$
from the nearest EP
singularity.

For the general and popular Hermitian Hamiltonians,
{\em none\,} of its EP singularities
$\lambda^{(EP)}$ can be real.
Some of them may be made real via
an analytic continuation of the model in $\lambda$.
Naturally,
what seems to be inadvertently lost is the self-adjointness of the
Hamiltonian.
Still, surprisingly enough,
there exist certain specific models in which
it is possible to recover the self-adjointness.
One of the best proofs of the existence of such an option
is even provided by the above-mentioned
quartic anharmonic model with  $V \sim x^4$.
In it, the non-Hermiticity of $H(\lambda)$ can most easily be achieved
via a counterintuitive choice of a negative $\lambda=-\kappa^2<0$.
Indeed, after
such a change of the sign, utterly unexpectedly,
the spectrum $\{E_n(-\kappa^2)\}$ is found real
again \cite{BB,BG,Jones}.

\subsection{Hiddenly Hermitian Hamiltonians}

In 1993, the occurrence of the latter surprise
was related, in
a marginal remark~\cite{BG}, to the
${\cal PT}-$symmetry of the non-Hermitian model. A few years later,
Bender with coauthors \cite{BB,Carl}
converted the remark into a new paradigm. An innovated,
more flexible formulation of
quantum theory has been born
admitting the existence
of unitary quantum evolution generated by non-Hermitian
Hamiltonians \cite{ali}.
The status of EPs has thoroughly
been changed. It evolved from the mere mathematical curiosity to
the concept of a central phenomenological interest \cite{book}.

There exist at least two keys
to the disentanglement of these
developments.
The main one makes use of
the correspondence
between the conventional Hermitian Hamiltonians
and their upgraded versions
called quasi-Hermitian \cite{Geyer},
pseudo-Hermitian \cite{ali},
crypto-Hermitian \cite{Smilga}, or
non-Hermitian but hermitizable \cite{Carl}.
Within this framework,
the second key to a correct description of the
EP-related quantum physics is more
pragmatic, aimed at a simplification of the
underlying mathematics. In this spirit,
several authors
suggested to
circumvent some of the technical obstacles
via an
{\it ad hoc}, apparently redundant
assumption like, say, a
pseudo-bosonization
of the operators of observables  \cite{Fabio}
(for a few reviews of some further possibilities
see, e.g., \cite{book,ATbook}).

The second key to
the necessary amendment of the conventional model-building
strategy is
model-dependent,
emphasizing the feasibility of the
calculations.
Its main tool lies in a maximal
technical
simplification of the operators.
As
a characteristic
implementation of such an idea
let us recall the
Bender's and Boettcher's
requirement \cite{BB}
of having the
Hamiltonians non-Hermitian but parity-time symmetric
(${\cal PT}-$symmetric).
Indeed, in effect,
the latter restriction
made the ${\cal PT}-$symmetric theory
popular even far beyond
its initial closed-system
applications \cite{Cham,Christodoulides,Carlbook}.

Incidentally, the widespread belief in
the consistency of the ${\cal PT}-$symmetric theory
has recently been opposed  \cite{Trefethen,Siegl},
reconfirming
the older
criticism of the
theory
by mathematicians
\cite{Dieudonne}.
Fortunately, a specific
strategy circumventing the criticism
can be found described in the
physics-oriented review paper \cite{Geyer}.
The recommendation lies in
the exclusive use of the operators of observables
which are bounded.
Such a recommendation will be
followed in our present paper.

From the point of view of many phenomenologically oriented
users, the innovative potential of the amended theory may
be best illustrated by the observation that the ``conventional''
choice (\ref{kateq}) of a Hamiltonian can be
complemented by a qualitatively different,
manifestly non-Hermitian toy model
 \be
 H^{(2)}(\lambda)
 = \left [\begin {array}{cc} -1&\lambda
 \\{}-\lambda &1
 \end {array}\right ]\,.
 \label{mama}
 \ee
Still, the easily obtained formula
$E_\pm(\lambda) = \pm \sqrt{1-\lambda^2}$
which defines eigenvalues implies that
they stay real and non-degenerate
whenever $\lambda \in {\cal D}=(-1,1)$.
In this interval of parameters,
matrix (\ref{mama}) becomes Hermitizable.
It can, therefore,
play the role of a closed-system
quantum Hamiltonian.
We can infer that one of
the main qualitative innovations is that
in model (\ref{mama})  the
distance $\varrho$ of the EPs from the
real axis of parameters $\lambda$ is zero (in the
conventional model of Eq.~(\ref{kateq})
we had $\varrho=1$). This means that
one can expect the emergence of multiple innovative dynamical
features of the system.
In particular, the observable system
can get
{\em arbitrarily close\,} to
its ``unphysical'' EP-mediated phase-transition extreme.

In the light of the
off-diagonal antisymmetry of model (\ref{mama}),
the matrix is, certainly,
{\em maximally\,}
non-Hermitian \cite{maximal}.
In this sense it can be treated as an inspiration of
an AHO
generalization
 \be
 H^{(AHO)}_{\rm (antisymmetric)}(\lambda)
 =\left[ \begin {array}{cccc}
  1&\lambda\,V_{0,1}(\lambda)
  &\lambda\,V_{0,2}(\lambda)&\ldots
  \\{}-\lambda\,V_{0,1}(\lambda)&3
  &\lambda\,V_{1,2}(\lambda)
 &\ddots
 \\{}-\lambda\,V_{0,2}(\lambda)&-\lambda\,V_{1,2}(\lambda)
 &5&\ddots
 \\{}\vdots&\ddots&\ddots&\ddots
 \end {array} \right]\,.
 \label{penasy}
 \ee
Naturally, the practical applicability of such an over-ambitious model
seems very much suppressed by the purely numerical nature of
the related predictions.

\subsection{Tridiagonality constraint and the
assumption of {\cal PT-}symmetry\label{ptsect}}

We saw that even some small and manifestly
non-Hermitian matrices could play
the role of a standard,
non-numerically tractable toy-model Hamiltonian.
In \cite{maximal,tridiagonal} such an observation was generalized and
extended to the whole family
 \be
 H^{(N)}_{\rm (tridiagonal)}(\lambda)
 =\left[ \begin {array}{ccccc}
  1&b_1(\lambda)
  &0&\ldots&0
  \\{}-b_1(\lambda)&3
  &b_2(\lambda)
 &\ddots&\vdots
 \\{}0&-b_2(\lambda)
 &5&\ddots&0
 \\{}\vdots&\ddots&\ddots&\ddots&b_{N-1}(\lambda)
 \\{}0&\ldots&0&-b_{N-1}(\lambda)&2N-1
 \end {array} \right]\,
 \label{pentrid}
 \ee
of the
{\em tridiagonal\,} real $N$ by $N$
matrix candidates for a
non-Hermitian but hermitizable
perturbed-harmonic-oscillator-type Hamiltonian.

In the first step of the construction an
inessential shift $ E \to E-N$ of the
origin of the energy scale was used to
transform the main diagonal written in the boxed-symbol form
$\fbox{1,3,5,\ldots,2N-1}$ into its centrally symmetric
version
$\fbox{1-N,3-N,\ldots,N-3,N-1}$. Such a shift renders
the matrix traceless.
In the second step the
(still excessively large) number of the variable
matrix elements was halved by a decisive
simplifying assumption of the symmetry of the
matrix
of perturbations
with respect to the second diagonal.
The latter simplification
was based on
the popularity of
${\cal PT}-$symmetry \cite{Cham}
in
a specific version
imposed upon the anharmonicity
(see \cite{tridiagonal} for a more extensive
complementary discussion).
Its implementation yielded the
final form
of tridiagonal AHO Hamiltonians
 \be
 H^{(N)}_{\rm (toy)}(\lambda)
 =\left[ \begin {array}{ccccc}
1-N&b_1(\lambda)
  &0&\ldots&0
  \\{}-b_1(\lambda)&3-N
  &\ddots
 &\ddots&\vdots
 \\{}0&-b_2(\lambda)
 &\ddots&b_2(\lambda)&0
 \\{}\vdots&\ddots&\ddots&N-3&b_{1}(\lambda)
 \\{}0&\ldots&0&-b_{1}(\lambda)&N-1
 \end {array} \right]\,.
 \label{pentoy}
 \ee
The main mathematical merit of
such a choice of the
benchmark lies
in a manifestly non-numerical form of the
Hermitization (see the details in \cite{metricsaho}).
In our present paper, an even more
important phenomenological merit of the choice
is to be seen in the availability of
the extreme, singular $\lambda=\lambda^{(EPN)}$ limits
of the Hamiltonians in closed forms.
Indeed, the sequence of the EP limits
obtained, in \cite{maximal}, at all $N$
is formed by elementary matrices
 \be
 H^{(2)}_{\rm (toy)}(\lambda^{(EP)})
 = \left [\begin {array}{cc} -1&1\\{}-1&1
 \end {array}\right
 ]\,,
 \ \ \ \ \ \
 H^{(3)}_{\rm (toy)}(\lambda^{(EP)})
  = \left [\begin {array}{ccc} -2&\sqrt{2}&0\\{}-\sqrt{2}&0
 &\sqrt{2}\\{}0&-\sqrt{2}&2\end {array}\right ],
 \label{epthisset}
 \ee
 \ben
 H^{(4)}_{\rm (toy)}(\lambda^{(EP)}) = \left [\begin {array}{cccc}
  -3&\sqrt{3}   &0  &0\\
 -\sqrt{3}&-1   &2  &0\\
  0&-2  &1 &\sqrt{3}\\
  0&0&-\sqrt{3}&3
 \end {array}\right ]\,,
 \ \ \ \ \ \
 H^{(5)}_{\rm (toy)}(\lambda^{(EP)})
 =\left [\begin {array}{ccccc} -4&2&0&0&0\\{}-2&
 -2&\sqrt{6}&0&0\\{}0&-\sqrt{6}&0&\sqrt{6}&0
 \\{}0&0&-\sqrt{6}&2&2\\{}0&0&0&-2& 4\end {array}\right ]\,
 ,
 \een
etc.

\section{Pentadiagonal Hamiltonians \label{specice}}

In a small vicinity of any one of the latter extremes
the character of dynamics is interesting, determined by the closeness of
the EPN-related ``catastrophe'' \cite{catast,catastb}.
Unfortunately, a more detailed analysis of the
quantum-catastrophic
scenario
leads to the major
disappointment:  in Eq.~(\ref{kwinde}),
all of the EPN limits (\ref{epthisset})  lead, at any $N$, to
{\em the
same},
non-clustering, obstinate
$K=1$ degeneracies.
For the {\em tridiagonal\,} $\lambda-$dependent
matrices (\ref{pentoy}),
in this manner,
the menu of the
eligible scenarios of the EPN-related
quantum phased transition is
too narrow.
The mathematically user-friendly
tridiagonality constraint proves phenomenologically
over-restrictive.

\subsection{Search for anomalous degeneracies}

%at $\lambda \to \lambda^{(EPN)}$}

In contrast to the ``standard''
exceptional points
with a minimal
geometric multiplicity $K=1$,
their $K>1$ alternatives may be called
``anomalous''.
The latter term was
conjectured, in \cite{anomalous}, as reflecting
the highly plausible one-to-one correspondence
between the
tridiagonality of the Hamiltonian and the $K=1$
form of its EPN-boundary
loss-of-the-observability limit.
In this sense the most natural candidates for an
``anomalous'' EPN-related dynamical regime
may be Hamiltonians which are pentadiagonal.
Unfortunately, the corresponding general ansatz
 \be
 H^{(N)}_{\rm (pentadiagonal)}(\lambda)
 =\left[ \begin {array}{cccccc}
  1-N&b_1(\lambda)&c_1(\lambda)
  &0&\ldots&0
  \\{}-b_1(\lambda)&3-N
  &b_2(\lambda)
 &\ddots&\ddots&\vdots
 \\{}-c_1(\lambda)&-b_2(\lambda)
 &\ddots&\ddots&c_2(\lambda)&0
 \\{}0&\ddots&\ddots&N-5&b_{2}(\lambda)&c_1(\lambda)
 \\{}\vdots&\ddots&-c_2(\lambda)&-b_{2}(\lambda)&N-3&b_1(\lambda)
 \\{}0&\ldots&0&-c_1(\lambda)&-b_{1}(\lambda)&N-1
 \end {array} \right]
 \label{epen}
 \ee
can only be analyzed by numerical methods.
For our present purposes, a further simplification is needed.
Thus, a non-numerical tractability of matrix (\ref{epen})
can be achieved, e.g., by an additional
assumption that $b_n(\lambda)=0$ at all $n$.
In order to see this clearly, the resulting, simplified Hamiltonian
may be partitioned,
 \be
 H^{(N)}_{\rm (pent. special)}(\lambda)
 =\left[ \begin {array}{c|c|c|c|c|c}
  1-N&0&c_1(\lambda)
  &0&\ldots&0
  \\ \hline0&3-N
  &0
 &\ddots&\ddots&\vdots
 \\ \hline-c_1(\lambda)&0
 &\ddots&\ddots&c_2(\lambda)&0
 \\ \hline0&\ddots&\ddots&N-5&0&c_1(\lambda)
 \\ \hline\vdots&\ddots&-c_2(\lambda)&0&N-3&0
 \\ \hline0&\ldots&0&-c_1(\lambda)&0&N-1
 \end {array} \right]\,.
 \label{cepez}
 \ee
This reveals that the matrix is equal to a direct sum of
the two decoupled tridiagonal matrices
 \be
 H^{(N)}_{\rm (component\ one)}(\lambda)
 =\left[ \begin {array}{cccc}
  1-N&c_1(\lambda)
  &0&\ldots
 \\ -c_1(\lambda)&5-N
 &c_3(\lambda)&\ddots
 \\ 0&-c_3(\lambda)&9-N
 &\ddots
 \\ \vdots&\ddots&\ddots&\ddots
 \end {array} \right]\,
 \label{apezerb}
 \ee
and
 \be
 H^{(N)}_{\rm (component\ two)}(\lambda)
 =\left[ \begin {array}{cccc}
  3-N&c_2(\lambda)
  &0&\ldots
 \\ -c_2(\lambda)&7-N
 &c_4(\lambda)&\ddots
 \\ 0&-c_4(\lambda)&11-N
 &\ddots
 \\ \vdots&\ddots&\ddots&\ddots
 \end {array} \right]\,.
 \label{bepeze}
 \ee
In detail, every matrix in question
is identified, uniquely, by its
main diagonal. Thus, the matrix of
Eq.~(\ref{cepez}) may be assigned the
boxed symbol
$\fbox{1-N,3-N,\ldots,N-1}$, etc.
Also the direct-sum decomposition
of the latter matrix can be abbreviated as
follows,
 $$
 \fbox{1-N,3-N,\ldots,N-1}=\fbox{1-N,5-N,9-N, \ldots} \oplus
 \fbox{3-N,7-N,11-N,\ldots}\,.
 $$
The last elements of the summands are not displayed because they
vary with
the parity of $N$.
After the explicit specification of
the parity of $N$
we
arrive at the following two conclusions.

\begin{lemma}
At the even matrix dimension $N=2J$,
the decomposition
of the pentadiagonal
sparse-matrix model (\ref{cepez})
into its
tridiagonal AHO components (\ref{apezerb}) and (\ref{bepeze})
only supports the two standard $K=1$ EPJ limits (\ref{epthisset}),
with the two different
energies $\eta_\pm = \pm 1 \neq 0$
in Eq.~(\ref{siesta}).
\label{lemma1}
\end{lemma}
\begin{proof}
The main diagonal
$\fbox{1-N,3-N,\ldots,N-1}= \fbox{1-2J,3-2J,\ldots,2J-1}$
of matrix (\ref{cepez}) does not contain a central zero.
This means that the central interval $(-1,1)$
is ``too short''. Its two elements $-1$ and $1$
will be distributed among both of the
components (\ref{apezerb}) and (\ref{bepeze}). In the
resulting direct sum
  $$
  \fbox{1-N,3-N,\ldots,N-1}= %\fbox{1-2J,3-2J,\ldots,2J-1}=
  \fbox{1-2J,5-2J, \ldots,2J-3}\oplus
  \fbox{3-2J,7-2J,\ldots,2J-1}\,,
  $$
for this reason, both of the components will be centrally asymmetric.
The search for an anomalous EP with $K=2$ fails.
Even after a successful $J$ by $J$ realization
of the two separate EPJ limits using building blocks (\ref{epthisset}),
the requirement~(\ref{siesta}) will offer
two different, incompatible
values $\eta_\pm = \pm 1$ of the
eligible limiting EP energies.
The direct-sum decomposition yields the two
non-anomalous $K=1$ EPs
of the same small order  $J=N/2$.
\end{proof}

\begin{lemma}
At odd $N=2J+1$, both of the
tridiagonal matrices (\ref{apezerb}) and (\ref{bepeze})
admit the respective realizations (\ref{epthisset})
of their EP limits.
The related energies coincide so that
the direct sum (\ref{cepez})
admits the anomalous
EPN limit with
geometric multiplicity two.
\label{lemma2}
\end{lemma}

\begin{proof}
We have
  \be
  \fbox{1-N,3-N,\ldots,N-1}= \fbox{-2J,2-2J,\ldots,2J}=
  \fbox{-2J,4-2J, \ldots,2J}\oplus
  \fbox{2-2J,6-2J,\ldots,2J-2}
  \label{ledva}
  \ee
so that out of
the central triplet of integers $(-2,0,2)$,
the doublet  $(-2,2)$ remains long enough to be
a component of one of the sub-boxes.
For this reason, their respective dimensions $J+1$ and $J$
are now different.
This is compensated by the central symmetry
of the summands which
implies
the coincidence of the
separate EP-related energies,
$\eta_\pm = \eta=0$.
In
the direct sum (\ref{cepez})
of these matrices
the respective EP(J+1) and EPJ limits
degenerate to a single, shared, manifestly anomalous
EPN limit. The ``$K-$tuple clusterization'' phenomenon
(\ref{kwinde}) takes place at $K=2$.
\end{proof}

Our two Lemmas may be read as an
empirical support of
the ``no-go'' conjecture of
Ref.~\cite{anomalous} which
attributed the long-lasting failures of
the trial-and-error
search for the anomalous, $K>1$ EPN degeneracies
to the underlying matrix-tridiagonality restriction.
Our Lemmas confirm
that the mere weaker, pentadiagonality constraint
or its further band-matrix generalizations
need not necessarily help too much.
At the same time, the model
has several methodical merits.
First, it shows that any attempted
classification
of the Hamiltonians considered
in their EPN extreme
must involve the anomalous cases with $K>1$.
Second, it offers an independent support
for an apparently arbitrary
but unexpectedly fortunate
choice of the specific, AHO-based
Hamiltonian-operator (sub)matrices.
Third, the pentadiagonality
facilitates the
decoupling
of the system into
two subsystems, an idea which inspires, immediately,
a generalization to the
$K-$component partitionings with $K>2$.
Fourth,
one should emphasize that
the study of the pentadiagonal models
may be perceived as paving the way towards
its full-matrix extensions,
especially at
$\lambda \approx  \lambda^{(EP)}$
where
there emerges a number of technical problems
ranging from the
numerical ill-conditioning difficulties \cite{bhgt6}
up to the complicated nature of the
perturbation-approximation tractability of the
stable and unitary
systems when occurring near the EPN
singularities \cite{without}.

Certain unusual physical
phenomena might emerge involving not only the
collapse of the system (via
a ``fall into the EPN singularity'')
but also, in opposite direction, the processes of an
escape from the
degeneracy resembling the Big Bang in cosmology \cite{BiBa}.
In both of these directions
there are
illustrative examples available in the literature.
{\it Pars pro toto\,} let us mention Ref.~\cite{Uwe}
where the authors showed that in a
complex Bose-Hubbard model the
many-bosonic system
can either
escape from, or fall in, the Bose-Einstein condensation singularity.

\subsection{Corridors of unitarity\label{4c3}}

In an overall
model-building strategy
one should feel aware of the deep difference between the
effective-operator
studies of resonances in open systems
(working with complex energies and, hence,
not considered here)
and the closed system models
in which the
spectrum of energies is assumed real.
In the latter cases
a
corridor of a possible unitary-evolution
could exist and
connect
the weakly-anharmonic (WA)  and
the strong-coupling (SC) dynamical regimes.

% at $\lambda \neq \lambda^{(EPN)}$
%plot({0,3*sqrt(4-g^2),-3*sqrt(4-g^2),2*sqrt(4-g^2),-2*sqrt(4-g^2),
%> 1*sqrt(4-g^2),-1*sqrt(4-g^2)},g=-2..0,
%axes=normal,color=black,tickmarks=[4,5]);
%Initial setup - small g, EPN limit - g=g...

The partitioning of our present specific
quantum Hamiltonians  (\ref{cepez}) may be expected to simplify
the construction of the corridors.
First of all, its perturbation specification
in the WA dynamical regime (in which
all of the
off-diagonal elements $c_n(\lambda)$ have to remain small)
would be routine. In the opposite extreme
of the strongly non-Hermitian SC domain, both of
the tridiagonal-matrix components (\ref{apezerb})
and (\ref{bepeze}) of the Hamiltonian
become {\em separately\,} tractable as small perturbations
of their  respective
SC EP limits
(\ref{epthisset}).
Thus,
as long as the matrix dimensions $N$ remain finite,
an approximate construction of the bound states becomes feasible
in both of the WA or SC perturbation regimes.
The available
perturbation
approximations will
determine, roughly at least,
the families of
the
admissible matrix elements $c_n(\lambda)$
for which the spectrum remains real and discrete,
i.e., for which the evolution of the quantum
system in question remains unitary and, hence, stable.

In the gap between the two comparatively small
WA- and SC-applicability
subintervals of $\lambda \leq \lambda^{(WA)}_{\max}
\ll \lambda^{(EP)}$
and  $\lambda \geq \lambda^{(SC)}_{\min}
\lessapprox \lambda^{(EP)}$
it will be more difficult to
guarantee the reality of the spectrum.
Fortunately, what we still know is that
for
$\lambda \in (\lambda^{(WA)}_{\max}, \lambda^{(SC)}_{\min})$
there always
exists a special unitarity-compatible
parametrization of matrix elements
applicable to both of the
{\em independent\,}
tridiagonal-matrix sub-Hamiltonians
(\ref{apezerb})
and (\ref{bepeze}) and
constructed in \cite{maximal,tridiagonal}.
This implies that in the light of Lemma \ref{lemma2}
a unitarity corridor
will exist,
in the anomalous $K=2$ case
based on
the direct sum
of these components, for all of
our
pentadiagonal ${\cal PT}-$symmeric
anharmonic-oscillator Hamiltonians (\ref{cepez})
with odd $N$.
In particular, in the $N=7$ exemplification
of our
antisymmetrically anharmonic pentadiagonal Hamiltonian (\ref{cepez})
 \be
 {H^{(7)}}(\lambda)=  \left[ \begin {array}{ccccccc}
  1&0&\sqrt {3}g&0&0&0&0
 \\\noalign{\medskip}0&3&0&\sqrt {2}g&0&0&0
 \\\noalign{\medskip}-\sqrt {
3}g&0&5&0&{2}g&0&0
\\\noalign{\medskip}0&-\sqrt {2}g&0&7&0&\sqrt {2}g&0
 \\\noalign{\medskip}0&0&-{2}g&0&9&0&\sqrt {3}g
 \\\noalign{\medskip}0&0&0&-
 \sqrt {2}g&0&11&0\\\noalign{\medskip}0&0&0&0&-\sqrt {3}g&0&13
 \end {array} \right]\,
 \label{tomograf}
 \ee
(where we returned, for the sake of clarity of physics, to the
initial, unshifted harmonic-oscillator energy scale),
the
$\lambda-$dependence
of the
off-diagonal matrix elements is
controlled by a single
function $g = g(\lambda)$.
Such a simplification implies that
the related Schr\"{o}dinger bound-state problem
 \be
 {H^{[7]}}(g)\,|
 {\psi_n}(g)\kt=E_n(g)\,| {\psi_n}(g)\kt\,
 \label{themo}
 \ee
is solvable exactly,
 $$
 E_0(g)=7\,,\ \ \ \ \
 E_{\pm 1}(g)
 =7\pm \sqrt {4-{g}^{2}}
 %=
%(11\mp {\frac {1}{2}}{g}^{2}\mp {\frac {1}{32}}{g}^{4}+O \left( {g}^{6}
% \right) )
$$
 $$
 E_{\pm 2}(g)
 =7\pm 2\,\sqrt {4-{g}^{2}}
 %=
%(11\mp {\frac {1}{2}}{g}^{2}\mp {\frac {1}{32}}{g}^{4}+O \left( {g}^{6}
% \right) )
\,,\ \ \ \
 E_{\pm 3}(g)
 =7\pm 3\,\sqrt {4-{g}^{2}}\,.
 %=
%(11\mp {\frac {1}{2}}{g}^{2}\mp {\frac {1}{32}}{g}^{4}+O \left( {g}^{6}
% \right) )
$$
The model exemplifies the system
in which there exists a corridor of unitarity
which connects the self-adjoint harmonic-oscillator dynamics
realized
at $g=0$
with the EP7 extreme where $g=2$.
Thus,
whenever the growth of $g(\lambda)$ does not deviate
too much from the linear function,
the bound-state energies remain real and
well separated
along a path connecting the WA and SC
ends of
the
open interval of values
$g=g(\lambda) \in (0,2)$.

\subsection{SC dynamical extreme}

In the
WA regime with
$\lambda  \leq \lambda^{(WA)}_{\max}$
the optional auxiliary (and, say, monotonously increasing)
function $g(\lambda)$ is to be kept
small. Then, the
anharmonicity will remain
easily tractable by the standard
Rayleigh-Schr\"{o}dinger perturbation methods.
Not only at our $N=7$ but also
in all of the models of a pentadiagonal
direct-sum type
with arbitrary odd~$N$.
Near the opposite SC boundary where $g \lessapprox 2$
the EP7 degeneracy (\ref{siesta})
is reached. We may introduce a new small
parameter $\kappa=\kappa(\lambda)\in (0,1)$,
redefine $g=\widetilde{g}(\kappa)=2\,(1-\kappa^2)$ and consider the
related (``tilded'') modification of our
spectral problem (\ref{themo}) with,
naturally, the same exact eigenvalues
written now
in a reparameterized, SC-friendly form
 $$
 \widetilde{E}_0(\kappa)=7\,,\ \ \ \ \
 \widetilde{E}_{\pm 1}(\kappa)
 =7 \pm 2\,\sqrt {-{\kappa}^{4}+2\,{\kappa}^{2}}
 \sim 7 \pm 2\,\sqrt {2}{\kappa}+ {\cal O}(\kappa^3)\,,
 $$
 %$$\approx (7+2\,\sqrt {2}\kappa-1/2\,\sqrt {2}{\kappa}^{3}-1/16\,\sqrt {2}{
%\kappa}^{5}-{\frac {1}{64}}\,\sqrt {2}{\kappa}^{7}+O \left( {\kappa}^{
%9} \right) )
% $$
  $$
  \widetilde{E}_{\pm 2}(\kappa)=7 \pm 4\,\sqrt {-{\kappa}^{4}
  +2\,{\kappa}^{2}}
 \,,\ \ \ \ \
  \widetilde{E}_{\pm 3}(\kappa)
  =7 \pm 6\,\sqrt {-{\kappa}^{4}+2\,{\kappa}^{2}}\,.
 $$
With $\kappa \in (0, 1)$ these values remain all real.
%.
%
%(4)
%picture: three EPs due to symmetry (artificial)
%
%. $ H^{[7]}(0)$ $ H^{[7]}(0)$
%
%plot({Jo7[1,1],Jo7[2,2],Jo7[3,3],(Jo7[4,4])
%> ,(Jo7[5,5]),Jo7[6,6],Jo7[7,7]},
%> kappa=-1.59..1.59,color=black,tickmarks=[3,3],axes=framed);
%

In the SC EP7 limit $\kappa \to 0$
the spectrum becomes degenerate and the
Hamiltonian ceases to be diagonalizable.
Even after the determination of
the EP value $\eta=7$ of the degenerate energy,
a modified eigenvalue problem
can be considered in this limit,
 \be
 H^{(7)}(\lambda^{(EP7)})\,Q^{[7]}=Q^{[7]}\,J^{[7]}(\eta)\,.
 \label{kano}
 \ee
In the light of Lemma \ref{lemma2}
we may immediately pick up the so
called canonical representation
 \be
 J^{[7]}(\eta)
 =\left[ \begin {array}{cccc|ccc}
  \eta&1&0&0&0&0&0
  \\{}0&\eta&1
 &0&0&0&0
 \\{}0&0&\eta&1&0&0&0
 \\{}0&0&0&\eta&0
 &0&0
 \\
 \hline
 0&0&0&0&\eta&1&0
 \\{}0&0&0&0&0&\eta&1
 \\{}0&0&0&0&0&0&\eta
 \end {array} \right]\,
 \label{euu}
 \ee
of our limiting toy-model Hamiltonian in (\ref{kano}).
Although such a choice is not unique
(see \cite{passage}
for some related mathematical comments and technicalities),
its $K=2$ direct-sum form (\ref{euu})
composed
of
the two Jordan-block matrices
is by far its
most popular version.

The partitioning in Eq.~(\ref{euu})
re-emphasizes that there is no mutual coupling between
the even and odd indices in our toy-model Hamiltonian
(\ref{tomograf}). Also in the EP7 limit,
matrix $H^{(7)}(\lambda^{(EP7)})$ may be interpreted
as a direct sum of the two smaller components,
$$
  H^{(7)}(\lambda^{(EP7)})=H^{[odd]}\bigoplus H^{[even]}
  $$
where
 $$
  H^{[odd]}=
 \left[ \begin {array}{cccc}
 1&2\,\sqrt {3}&0&0
  \\-2\, \sqrt {3}&5&4&0
 \\0&-4&9&2\,\sqrt {3}
 \\0&0
 &-2\,\sqrt {3}&13
 \end {array} \right]\,,
 \ \ \ \ \
  H^{[even]}=
 \left[ \begin {array}{ccc}
 3&2\,\sqrt {2}&0
  \\-2\, \sqrt {2}&7&2\,\sqrt {2}
 \\0
 &-2\,\sqrt {2}&11
 \end {array} \right]
 \,.
 $$
Another
specific illustrative merit
of our $N=7$ model
lies
the
availability of the explicit  transition-matrix
solution
 $$
 Q^{[7]}= \left[ \begin {array}{ccccccc} -48&24&-6&1&0&0&0\\{}0
&8&-4&1&8&-4&1\\{}-48\,\sqrt {3}&16\,\sqrt {3}&-2\,
\sqrt {3}&0&0&0&0\\{}0&8\,\sqrt {2}&-2\,\sqrt {2}&0&8
\,\sqrt {2}&-2\,\sqrt {2}&0\\{}-48\,\sqrt {3}&8\,
\sqrt {3}&0&0&0&0&0\\{}0&8&0&0&8&0&0
\\{}-48&0&0&0&0&0&0\end {array} \right]
 $$
of Eq.~(\ref{themo}).
Indeed, in the constructions of the SC
perturbation series
the transition matrices
can be treated as
an optimal
substitute for the
non-existent complete
sets of unperturbed eigenvectors.
In this sense
the role of
an unperturbed SC basis
is played by the transition matrices themselves
(see, e.g., \cite{admissible}).

%\section{Multidiagonal anharmonicities and anomalous EPNs }

\section{Maximally anharmonic full-matrix models \label{specide}}

The simulation of quantum dynamics near EPNs
using multidiagonal matrix Hamiltonians
was proposed in paper \cite{anomalous}.
The scope of the study was restricted there,
due to the apparently purely numerical nature
of the problem, to
the smallest matrix dimensions $N \leq 6$.
The decision was a bit unfortunate
because
in the light of subsection \ref{4c3}
and
Lemmas \ref{lemma1} and \ref{lemma2}
above, the next option with $N=7$
would have been perceivably more instructive.
Still, the message
delivered
by Ref.~\cite{anomalous} remained significant, showing that
the search for anomalous
EPN singularities
with optional
geometric multiplicities $K$
should be based on a systematic study of
non-tridiagonal, multi-diagonal matrix models.

On such a background
we are prepared to
make the next step towards the construction of a complete
list of the first few
AHO-based phase-transition scenarios.
Our considerations will start from
the entirely general real $N$ by $N$ matrix ansatz
 \be
 H^{(N)}_{\rm (full)}(\lambda)
 =\left[ \begin {array}{cccccc}
  1-N&b_1(\lambda)&c_1(\lambda)
  &d_1(\lambda)&\ldots&\omega_1(\lambda)
  \\{}-b_1(\lambda)&3-N
  &b_2(\lambda)
 &c_2(\lambda)&\ddots&\vdots
 \\{}-c_1(\lambda)&\ddots
 &\ddots&\ddots&\ddots&d_1(\lambda)
 \\{}-d_1(\lambda)&\ddots&-b_3(\lambda)
 &N-5&b_{2}(\lambda)&c_1(\lambda)
 \\{}\vdots&\ddots&-c_2(\lambda)&-b_{2}(\lambda)&N-3&b_1(\lambda)
 \\{}-\omega_1(\lambda)&\ldots&-d_1(\lambda)
 &-c_1(\lambda)&-b_{1}(\lambda)&N-1
 \end {array} \right]\,
 \label{pefull}
 \ee
carrying, in the diagonally dominated cases at least, the
interpretation of an acceptable physical Hamiltonian.
In the language of reviews~\cite{Carl,ali,Geyer,MZbook}
such a Hamiltonian
generates, in
the corresponding quasi-Hermitian Schr\"{o}dinger picture,
the standard unitary evolution of a closed and stable
quantum system
with the states living in an $N-$dimensional physical
Hilbert space ${\cal H}^{(N)}$.

\subsection{Classification of EPNs using generalized boxed symbols}

What lies in the very center of our attention
are the EPN-degeneracy requirements (\ref{siesta})
and (\ref{kwinde})
in application to the
general ${\cal PT}-$symmeric
anharmonic-oscillator Hamiltonian (\ref{pefull}).
We will assume that the real and antisymmetric-matrix
anharmonicity
(i.e., in a way, a maximally non-Hermitian anharmonicity)
vanishes in the
unperturbed-harmonic-oscillator limit of $\lambda \to \lambda^{(HO)}=0$.

After the model leaves the weakly anharmonic dynamical
regime controlled by the diagonally dominated Hamiltonian,
one has to start to fine-tune the $\lambda-$dependence of all of the
off-diagonal matrix elements in order to guarantee the
survival of the reality and non-degeneracy
of the bound-state energy spectrum,
i.e., of the
stability and unitarity of the evolution
in a ``physical'' interval of $\lambda \in (0,\lambda^{(\max)})$.
This opens the theoretically as well as
phenomenologically highly attractive
possibility of the existence of such a set of the off-diagonal
matrix elements (i.e., of the {\it ad hoc\,}
$\lambda-$dependent anharmonicities) that the ultimate
loss of the observability of the system
(i.e., the loss of the reality and non-degeneracy of
the energy spectrum at $\lambda^{(\max)}$)
would have a very specific
complete-degeneracy
EPN form as prescribed by
Eq.~(\ref{siesta}) above.
This would mean that the system can reach
such a strong-perturbation regime
with
$\lambda \approx \lambda^{(\max)}=\lambda^{EP)}$ where
the system reaches its EPN degeneracy as characterized by
Eq.~(\ref{esist}).

Our project of the study of such a possibility
was inspired by the recent paper \cite{anomalous}
in which it has been emphasized
(and, via a few
{\it ad hoc\,} $N=6$ examples, demonstrated) that
the
complete EPN degeneracy of the spectrum
may be ``anomalous'', accompanied by the mere $K-$clustered,
geometric multiplicity reflecting
degeneracy
of the eigenstates
as described by Eq.~(\ref{kwinde}).
We are going to complement Ref.~\cite{anomalous}
by extending its purely numerical $K=2$ and $K=3$ samples
of the EP6-related
Hamiltonians to
a universal non-numerical construction of the EPN-related models
characterized by the arbitrary integer matrix dimensions $N\geq 2$
and by the arbitrary integer geometric multiplicites $K \geq 1$.

The idea of our constructions is threefold.
Firstly, in a way inspired by the specific $K=2$
pentadiagonal-matrix constructive results of Lemma \ref{lemma2} above,
the full-matrix
Hamiltonian (\ref{pefull}) will again be characterized,
without any danger of confusion,
by the slightly generalized version of its identification
using the
description of its main diagonal
via the same boxed
left-right symmetric symbol
 \be
 {\cal S}(N,1)=\fbox{1-N,3-N,\ldots,N-3,N-1}\,
 \label{bosy}
 \ee
(see Definition \ref{bece} in Appendix A below). Secondly,
in the same spirit,
we will search for a multiterm (i.e., more precisely, $K-$term)
generalization of the
direct-sum $K=2$ expansion
(\ref{ledva}).
Thirdly,
we will emphasize that or present physics-related,
AHO-based
requirement of the
non-numerical tractability of our toy-model EP limits
is in a one-to-one correspondence with the
mathematics-related requirement
that
all of the components of
our innovative and systematic $K-$term direct-sum expansions
must remain centrally symmetric.

\subsection{Enhancement of clusterization $K \longrightarrow K+1$}

Some of the combinatorial
aspects of the direct-sum expansions
are clarified in Appendix A. In a more pragmatic spirit
let us return to the concept of the partitioning
which proved
to be useful in the pentadiagonal-matrix case above.
On a way toward its $K>2$ generalizations
let us first concentrate, for introduction, on a special form of
the most elementary partitioning-motivated
separation of the outer rows and columns followed by
their reduction to the mere two non-vanishing elements. This yields
the matrix
 \be
 H^{(N)}_{\rm (spec.partit.)}(\lambda)
 =\left[ \begin {array}{c|cccc|c}
  1-N&0&0
  &\ldots&0&\omega_1(\lambda)
  \\ \hline 0&3-N
  &b_2(\lambda)
 &\ldots&z_2(\lambda)&0
 \\ 0&-b_2(\lambda)
 &\ddots&\ddots&\vdots&0
 \\ \vdots&\vdots&\ddots&N-5&b_{2}(\lambda)&\vdots
 \\ {} 0&-z_2(\lambda)&\ldots&-b_{2}(\lambda)&N-3&0
 \\ \hline -\omega_1(\lambda)&0&0&\ldots&0&N-1
 \end {array} \right]
 \label{pepart}
 \ee
which can be perceived as a direct sum of the two
decoupled and fully independent
smaller matrices,
 $$
 H^{(N)}_{\rm (spec.partit.)}(\lambda)
 =
 \left [(N-1)\times H^{(2)}_{\rm (toy)}(\lambda)\right ]
 \oplus H^{(N-2)}_{\rm (full)}(\lambda)\,.
 $$
Such a decomposition
preserves the possible AHO-related nature of both of the components.
In our present abbreviated notation of Appendix A
this means that
the centrally symmetric main diagonal
of the left-hand-side Hamiltonian
(assigned the
appropriate boxed symbol)
can be reinterpreted as the following composition of
the two other AHO-representing symbols,
 $$
 \fbox{1-N,3-N,\ldots,N-1}=
 \fbox{1-N,N-1}\,  \oplus  \fbox{3-N,5-N,\ldots,N-3}\,.
 $$
Immediately, the trick leading to the
latter direct-sum decomposition can be generalized to
any centrally symmetric partitioning.
At a few illustative
matrix dimensions $N$
the complete lists of these
decompositions may be found listed,
using a slightly simplified notation,
in Appendix~A.

\section{Maximally anharmonic  sparse-matrix models\label{specife}}

Partitioning of the AHO-representing centrally symmetric main diagonal
${\cal S}(N,1)$ of Eq.~(\ref{bosy}) into $K-$plets
of its shorter,
centrally symmetric
equidistant subsets
 \be
 {\cal S}(M,L)=\fbox{(1-M)\,L,(3-M)\,L,\ldots,(M-3)\,L,(M-1)\,L}\,
 \label{hjugobossy}
 \ee
forms our main model-building principle.
In its spirit, every AHO-based Hamiltonian
is to be represented
as a direct sum of
its AHO-based sub-Hamiltonian building blocks.
At the first few dimensions $N$,
the systematic
application
of this recipe is to be illustrated in
what follows.

\subsection{The choice of $N=2$ and $N=3$: no anomalous degeneracies}

For our present AHO class of
$\lambda-$dependent Hamiltonians (\ref{pefull})
there exists strictly one,
unique EP2 limit
satisfying our restrictions at $N=2$, namely, the
matrix $H^{(2)}_{\rm (toy)}(\lambda^{(EP)})$
as displayed in Eq.~(\ref{epthisset}).
In the abbreviated notation using the
centrally symmetric
boxed symbols such a matrix is characterized by
${\cal S}(2,1)=\fbox{-1,1}$. The number $a(N)$ of eligible
scenarios is one, $a(2)=1$.

Similarly, at $N=3$ we have
$a(3)=1$ and
the unique
limit $H^{(3)}_{\rm (toy)}(\lambda^{(EP3)})$
represented by the boxed $K=1$
symbol ${\cal S}(3,1)=\fbox{-2,0,2}$ and
by the matrix
displayed in Eq.~(\ref{epthisset}).

\subsection{The simplest anomalous case with $N=4$ and $K=2$}

Besides the trivial $K=1$ option
with symbol $\fbox{-3,-1,1,3}$, there exists strictly one other
centrally symmetric possibility of decomposition
  $$
 \fbox{-3,-1,1,3}=\fbox{-1,1} \oplus \fbox{-3,3}
\,,\ \ \ \ K=2 \,.
 $$
In the limit $\lambda \to \lambda^{(EP4)}$
this direct sum represents the seven-diagonal but very
sparse AHO matrix
\be
 H^{(4)}_{(K=2)}(\lambda^{(EP4)}) = \left [\begin {array}{rrrr}
  -3&0   &0  &3\\
 0&1   &-1  &0\\
  0&-1  &1 &0\\
 -3&0&0&3
 \end {array}\right ]\,.
 \ee
In the unitarity domain where $\lambda \neq \lambda^{(EP4)}$
the number of the eligible dynamical scenarios is two, $a(4)=2$.

\subsection{Two $K=2$ options at $N=5$}

At $N=5$ the number of scenarios is three, $a(5)=3$.
Indeed, besides the trivial case, we have
the two $K=2$ decompositions
  $
 \fbox{-4,-2,0,2,4}=\fbox{-2,0,2} \oplus \fbox{-4,4}
 $ and
 $
 \fbox{-4,-2,0,2,4}=\fbox{-4,0,4} \oplus \fbox{-2,2} $,
representing the two respective EP5 limiting matrices, viz., the
nine-diagonal
 $$
  H^{(5)}_{(K=2,a)}(\lambda^{(EP5)})
 =\left [\begin {array}{rrrrr}
  -4&0&0&0&4
 \\{}0& -2&\sqrt{2}&0&0
 \\{}0&-\sqrt{2}&0&\sqrt{2}&0
 \\{}0&0&-\sqrt{2}&2&0
 \\{}-4&0&0&0& 4
 \end {array}\right ]\,
 $$
and the pentadiagonal
 $$
  H^{(5)}_{(K=2,b)}(\lambda^{(EP5)})
 =\left [\begin {array}{crccc}
    -4&0&2\sqrt{2}&0&0
 \\{}0& -2&0&2&0
 \\{}-2\sqrt{2}&0&0&0&2\sqrt{2}
 \\{}0&-2&0&2&0
 \\{}0&0&-2\sqrt{2}&0& 4
 \end {array}\right ]\,.
 $$
The latter matrix
fits in the classification pattern as provided by
Lemma~\ref{lemma2} above.

\subsection{The first occurrence of $K=3$ at $N=6$}

Besides the trivial, non-degenerate EP6 limit
with $K=1$ we have to consider its anomalous descendants, viz., the
unique $K=2$ decomposition
  $
 \fbox{-5,-3,-1,1,3,5}=\fbox{-3,-1,1,3} \oplus \fbox{-5,5}\,
 $
and the
unique $K=3$ decomposition
  $
 \fbox{-5,-3,-1,1,3,5}= \fbox{-1,1}
  \oplus \fbox{-3,3} \oplus \fbox{-5,5}
 $.
In both of these cases the
necessary direct-sum components of
$H^{(6)}_{(K=2,K=3)}(\lambda^{(EP6)})$
may be found displayed in Eq.~(\ref{epthisset}).
In the latter case, for example, we get
 $$
  H^{(6)}_{(K=3)}(\lambda^{(EP6)})
 =\left [\begin {array}{rrrccc}
    -5&0&0&0&0&5
 \\{}0& -3&0&0&3&0
 \\{}0&0&-1&1&0&0
 \\{}0&0&-1&1&0&0
 \\{}0&-3&0&0&3&0
 \\{}-5&0&0&0& 0&5
 \end {array}\right ]\,.
 $$
Summarizing, the number of scenarios is $a(6)=3$.
Incidentally,
the role
and consequences of small perturbations of the latter matrix have
numerically been
analyzed in \cite{anomalous}.

\subsection{Paradox of decrease of $a(N)$ between $N=7$ and $N=8$}

At $N=7$ the number of the eligible EP7 scenarios is $a(7)=6$
because
the usual trivial $K=1$ option can be accompanied by
the following quintuplet of anomalous EP7 direct sums,
  $$
 \fbox{-6,-4,-2,0,2,4,6}=
  \fbox{-4,-2,0,2,4} \oplus \fbox{-6,6}\,,
 \ \ \ \ \ K=2\,,
  $$
%and
  $$
 \fbox{-6,-4,-2,0,2,4,6}=\fbox{-2,0,2} \oplus \fbox{-4,4}
 \oplus \fbox{-6,6}\,,
 \ \ \ \ \ K=3\,,
 $$
%and
  $$
 \fbox{-6,-4,-2,0,2,4,6}=\fbox{-4,0,4} \oplus \fbox{-2,2}
 \oplus \fbox{-6,6}\,,
 \ \ \ \ \ K=3
 $$
%and
  $$
 \fbox{-6,-4,-2,0,2,4,6}=
 \fbox{-4,0,4} \oplus \fbox{-6,-2,2,6}\,,
 \ \ \ \ \ K=2\,,
 $$
%and
  $$
 \fbox{-6,-4,-2,0,2,4,6}=\fbox{-6,0,6} \oplus \fbox{-2,2}
 \oplus \fbox{-4,4}\,,
 \ \ \ \ \ K=3\,.
 $$
In contrast, at
$N=8$
we have $a(8)=4$, i.e.,
only the triplet of the anomalous, $K>1$ direct sums becomes available,
viz.,
  $$
 \fbox{-7,-5,-3,-1,1,3,5,7}=\fbox{-5,-3,-1,1,3,5}
  \oplus \fbox{-7,7}\,,
 \ \ \ \ \ K=2\,,
 $$
%and
  $$
 \fbox{-7,-5,-3,-1,1,3,5,7}=\fbox{-3,-1,1,3} \oplus \fbox{-5,5}
 \oplus \fbox{-7,7}\,,
 \ \ \ \ \ K=3\,,
 $$
and, finally, the first four-term direct-sum decomposition
  $$
 \fbox{-7,-5,-3,-1,1,3,5,7}= \fbox{-1,1} \oplus \fbox{-3,3}
  \oplus \fbox{-5,5}\,,
  \oplus \fbox{-7,7}\,,
 \ \ \ \ \ K=4\,
    $$
representing a
fifteen-diagonal but very sparse matrix
$ H^{(8)}_{(K=4)}(\lambda^{(EP8)})$.

\section{Discussion \label{speciae}}

The concept of the exceptional-point
value $\lambda^{(EP)}$ of a real parameter
$\lambda \in \mathbb{R}$ in a linear operator
$H(\lambda)$
proved, originally, useful
just in mathematics \cite{Kato,book}.
Physics behind the EPs
remained obscure.
The situation has changed after several authors
discovered
that the concept admits
applicability in multiple branches of quantum as well as non-quantum
physics \cite{Carl,Christodoulides,Carlbook}.
{\it Pars pro toto\,} let us mention that
in the subdomain of quantum physics the values of $\lambda^{(EP)}$
acquired the status
of instants of an experimentally realizable
quantum phase transition \cite{BB,Geyer,denis}.
The
boundary-of-stability
role played by the values of $\lambda^{(EP)}$
attracted, therefore, attention of experimentalists \cite{Heiss,Joglekar}
as well as of theoreticians \cite{Trefethen,Dorey}.

\subsection{Unitary vs. non-unitary quantum systems}

Whenever one keeps the evolution
unitary, the values of $\lambda^{(EP)}$
mark the
points of the loss of the observability
of the system \cite{catast,sdenisem}.
In this sense, our present
hierarchy of specific AHO models may be
treated as certain
exactly solvable quantum analogue of the Thom's  typology of classical
catastrophes \cite{Zeeman},
with potential applicability to closed as well as open systems.

In our present paper we were exclusively
interested in the former type of applications.
We showed that
the dynamics
of any unitary quantum system
(i.e., typically, its stability with respect to small perturbations)
is particularly strongly influenced by the
EPs. We should only add that an extreme care must be paid
to
the Stone theorem \cite{Stone}
requiring the
Hermiticity of
$H(\lambda)$
in the related physical Hilbert space ${\cal H}$.
This means that
a hermitization of the Hamiltonian is needed
\cite{Geyer}.
Such a process usually involves a
reconstruction of an
appropriate amended inner product
in the conventional but
manifestly unphysical Hilbert space ${\cal K}$.
Interested readers may find a sample of
such a reconstruction of ${\cal H}$,
say, in \cite{scirep}.

In the realistic models one often
encounters a paradox that the construction of
the correct, amended inner product
may happen to be prohibitively complicated, i.e.,
from the pragmatic point of view,
inaccessible.
This is the reason why people often
postpone the problem and use,
temporarily, a simplified inner product.
For our present, user-friendly, matrix-represented
AHO Hamiltonians of closed systems with $N < \infty$
such a purely technical obstacle does not occur.
The
reconstruction  of ${\cal H}$
(left to the readers)
would be a routine
application of linear algebra.
Reclassifying the real AHO matrices which are
non-Hermitian in our auxiliary space
${\cal K}=\mathbb{R}^N$ (that's why
we write $H \neq H^\dagger$)
into operators which are, by construction, Hermitian in ${\cal H}$
(in \cite{MZbook}, for example, we wrote $H=H^\ddagger$).

Let us add that in the more common studies
of effective non-Hermitian Hamiltonians $H_{\rm eff}(\lambda)$
describing the manifestly non-unitary evolution of the
open quantum systems
the role of the physical space is transferred back
from the ``unfriendly'' ${\cal H}$
to the ``original'' ${\cal K}$.
This means that
the difficult
construction of a nontrivial
inner product
is not needed.
Thus, it is not too surprising that
in the framework of the quantum theory of open systems
the use of various non-Hermitian Hamiltonians
becomes increasingly popular
\cite{Christodoulides}.
The study and/or the search for the
manifestations of the presence and structure of
the open-system EPs
enters, due to their simpler forms,
many new theoretical as well as phenomenological
territories \cite{Carlbook,stoplight}.

\subsection{Conclusions}

In our present paper we proposed one of the
potentially most useful classifications of the unitary quantum
processes of phase transitions mediated by the Kato's
exceptional points. Emphasis
has been put upon the exact solvability
of the underlying, AHO-based benchmark models.
One of the main merits of our
constructive classification is that the exact solvability
of these $N$ by $N$ matrix loss-of-observability
benchmark models
is guaranteed at an arbitrarily large matrix dimension $N$.

From the point of view of physics,
the success of our construction
appeared to be a consequence
of our fortunate initial choice of the family of
the phenomenological real-matrix Hamiltonians which
combined the harmonic-oscillator-inspired
equidistance of their diagonal matrix elements with the
entirely general but extreme,
maximally non-Hermitian (i.e., antisymmetric)
but ${\cal PT}-$symmetric
full-matrix form of the anharmonicities.
In such a setup
our search for a complete set of the EPN-supporting
benchmark models
which would remain non-numerical
at any $N$ was entirely straightforward.
A constructive
classification
of the closed quantum systems
living in a vicinity of their
EPN-mediated phase transition has been achieved.
We believe that in the nearest future
the direct-sum nature of our models
may be expected to facilitate the task of a guarantee of
stability with respect to small perturbations
at larger $N$ and $K$.

The most welcome serendipitious byproduct of our construction
is two-fold. First,
we managed to overcome the widespread belief that
in the study of the EPN-related phenomena
the techniques
(as sampled, say, in Ref.~\cite{anomalous})
must necessarily remain
numerical and
restricted to the systems with
the smallest matrix dimensions $N$.
Second,
we found that an increase of
the number of diagonals
in $H^{(N)}(\lambda)$ does not play any significant role.
We have shown,
in particular, that
the next-to-tridiagonal (i.e.,
pentadiagonal) choice of
matrices does not help too much.
We may conclude that the number of diagonals
is inessential, and that at $N$ as
small as six, and for the geometric multiplicity
as small as $K=3$, one needs
as many as eleven
diagonals
to simulate the anomalous
behavior of the system near its EP6 singularity.

In the language of mathematics
we were forced to introduce a certain new version of
the notion of partitioning. Starting from the
idea of the uniqueness of a simplified representation
of our specific matrices by their main diagonals
we took into consideration the
availability of tridiagonal building-block
sub-Hamiltonians. Then we
added the concept of a direct sum
of these AHO-based sub-Hamiltonians
forming a
decomposition of the
main $N$ by $N$ Hamiltonian.
This led us to the development of a
special partitioning technique which is described
in Appendix A. Along these lines we arrived,
at an arbitrary preselected matrix dimension $N$,
at
an exhaustive characterization and ``numbering''
of all of the eligible
EPN-related setups.
In a final step,
this result opened the way to the guarantee of the
completeness of the
classification.

We can summarize that
in the context of the study
and classification of quantum systems
in a small vicinity
of the instant of phase transition,
our key idea of decomposition of
certain toy-model $N$ by $N$ matrix Hamiltonians
into direct sums of
their AHO-matrix components
proved efficient and productive.
It was shown to lead to an exhaustive list of
multidiagonal AHO-related matrices $H^{(N)}(\lambda)$
forming an infinite family of
benchmark EPN-supporting Hamiltonians.
These representatives of the systems near a collapse
happened to be
unexpectedly user-friendly, especially due to their
characteristic sparse-matrix
structure.

\section*{Appendix A: Non-equivalent EPN scenarios}

In the present AHO-based phenomenological models
numbered by the Hamiltonian-matrix dimensions $N=(1), 2,3,\ldots$,
the counts of the eligible non-equivalent
EPN-related dynamical scenarios
form a sequence
 \be
 a(N) = (0), 1, 1, 2, 3, 3, 6, 4, 11, 6, 17, 7, 32, 8, 39, 13, 40,
  \ldots\,.
 \ee
The evaluation of the sequence
is important and useful for at least two physical reasons. First,
beyond the smallest $N$,
it enables us to check the completeness
of the EPN-related
dynamical alternatives.
Second,
the asymptotically exponential growth of the sequence
indicates that at the larger $N$s, the
menus of the EPN-supporting toy models will be dominated
by the anomalous,  multidiagonal $K>1$ Hamiltonians.

Besides that, the properties of the sequence
are of an independent mathematical interest.
First of all we notice that our sequence
seems composed of the two apparently simpler,
monotonously increasing integer subsequences.
They have to be discussed separately.

\subsection*{A.1. Subsequence of $a(N)$ with even $N=2J$, $J=1,2,\ldots$.}

The values of the subsequence
 \be
 b(J)=a(2J) =  1, 2, 3, 4, 6, 7, 8, 13, 14, 15, 25, 26, 33, 50, \ldots
 \ee
may be generated by the algorithm described in \cite{oeiseven}.
Recalling this source
(available, temporarily, under the identification number
A336739)
let us summarize a few key
mathematical features of the (sub)sequence.

\begin{defn}
\label{ace}
The quantity b(n) is the number of	
decompositions of {B}(n,1) into disjoint unions of {B}(j,k)
where {B}(j,k) is the set of numbers \{(2\,i-1)\,(2\,k-1), 1
$\leq$ i $\leq$ j \}.
\end{defn}

%DATA	
%1, 2, 3, 4, 6, 7, 8, 13, 14, 15, 25, 26, 33, 50, 51, 52,
% 95, 152, 153, 295, 296, 297, 542, 543, 672, 1329, 1330, 2055,
%4093, 4094, 4095, 6992, 10697, 10698, 21375, 21376, 21377, 39051,
%55948, 55949, 86454, 86455, 130396, 260765, 260766, 365839,
%731649, 1100442

 \noindent
It may be instructive to display a few examples:
	
{B}(n,1) are the sets \{1\}, \{1,3\}, \{1,3,5\}, \{1,3,5,7\}, ...,

{B}(n,2) are the sets \{3\}, \{3,9\}, \{3,9,15\}, \{3,9,15,21\}, ...,

{B}(n,3) are the sets \{5\}, \{5,15\}, \{5,15,25\}, \{5,15,25,35\}, ...,

 \noindent
etc. Thus, one can conclude that
there are two decompositions of {B}(2,1) = \{1,3\}, viz.,
trivial {B}(2,1) and nontrivial {B}(1,1) + {B}(1,2) = \{1\} + \{3\}.
Similarly, the complete list of the $a(5) = 6$
decompositions of \{1,3,5,7,9\}
is as follows:

  \{\{1,3,5,7,9\}\},

  \{\{1,3,5,7\}, \{9\}\},

  \{\{1,3,5\}, \{7\}, \{9\}\},

  \{\{1,3\}, \{5\}, \{7\}, \{9\}\},

  \{\{1\}, \{3\}, \{5\}, \{7\}, \{9\}\},

  \{\{3,9\}, \{1\}, \{5\}, \{7\}\}.

 \noindent
We should add that
the
notation used in definition \ref{ace} of quantities  {B}(j,k) is
mathematically optimal.
For the purposes of our present paper, nevertheless, it is
necessary to recall the equivalence of every {B}(j,k)
to one of the present boxed symbols. For example,
in place of
 {B}(3,1) = \{1,3,5\} we should write ${\cal B}$[3,1] =
  \fbox{-5,-3,-1,1,3,5}, etc.
Thus, definition \ref{ace} could be modified as follows.

\begin{defn}
\label{bece}
The quantity b(n) is the number of different decompositions
of ${\cal B}$[n,1] into
unions of ${\cal B}$[j,k] where ${\cal B}$[J,K] is defined as
the boxed symbol \\
. \hspace{2cm}
 \fbox{(2K-1)(1-2J),(2K-1)(3-2J),(2K-1)(5-2J), ... ,(2K-1)(2J-1)}.
\end{defn}

\subsection*{A.2.  Subsequence of $a(N)$ with odd $N=2J+1$, $J=1,2,\ldots$.}

%, under number, A335631,

The values of the subsequence
 \be
 c(J)=a(2J+1) =   1, 3, 6, 11, 17, 32, 39, 40, 56,  \ldots
 \ee
may be found discussed
 in \cite{oeisodd}.
 Using this source
let us summarize a few key aspects of this sequence
which carries, temporarily, the identification number A335631.

\begin{defn}
The quantity c(n) is the number of decompositions of {C}(n,1) into
disjoint unions of {C}(j,k) and G(q,r) where {C}(j,k) is the set of
numbers \{i\,k, 0 $\leq$ i  $\leq$ j \} and where G(q,r) is the
set of numbers \{(2\,p-1)\,r, 1  $\leq$ p  $\leq$ q \}.
\end{defn}

 \noindent
In a more explicit manner let us point out that
	
{C}(n,1) are the sets \{0,1\}, \{0,1,2\}, \{0,1,2,3\}, \{0,1,2,3,4\}, ...,

{C}(n,2) are the sets \{0,2\}, \{0,2,4\}, \{0,2,4,6\}, \{0,2,4,6,8\}, ...,

{C}(n,3) are the sets \{0,3\}, \{0,3,6\}, \{0,3,6,9\}, \{0,3,6,9,12\}, ...,

 \noindent
etc., and that

G(n,1) are the sets \{1\}, \{1,3\}, \{1,3,5\}, \{1,3,5,7\}, ...,

G(n,2) are the sets \{2\}, \{2,6\}, \{2,6,10\}, \{2,6,10,14\}, ...,

G(n,3) are the sets \{3\}, \{3,9\}, \{3,9,15\}, \{3,9,15,21\}, ...,

 \noindent
etc. We can say that
$a(2) = 3$ because the decompositions of {C}(2,1) = \{0,1,2\}
involve not only the trivial copy {C}(2,1) but also the nontrivial
formulae
{C}(1,2) + G(1,1) = \{0,2\} + \{1\} and
{C}(1,1) + G(1,2) = \{0,1\} + \{2\}.
Similarly: why $a(3) = 6$?
Because the decompositions of \{0,1,2,3\} are as follows:

  \{\{0,1,2,3\}\},

  \{\{0,1,2\}, \{3\}\},

  \{\{0,1\}, \{2\}, \{3\}\},

  \{\{0,2\}, \{1,3\}\},

  \{\{0,2\}, \{1\}, \{3\}\},

  \{\{0,3\}, \{1\}, \{2\}\}.

 \noindent
The one-to-one correspondence and the translation of this notation
to our present boxed-symbol language is again obvious, fully
analogous to
the one described in the preceding Appendix A.1.

%\end{verbatim}

\end{document}